\documentclass[final]{autart}
\usepackage{natbib}

\usepackage{tikz}
\usepackage{pgfplots}
\usetikzlibrary[pgfplots.groupplots]
\usetikzlibrary[patterns]
\usetikzlibrary[arrows.meta]
\usetikzlibrary[math]
\pgfplotsset{compat = 1.9}
\usepackage[shortlabels]{enumitem}
\usepackage{datetime}  %% FOR ADDING DATE

\newcommand{\mi}{\mathrm{i}}

\usepackage{wrapfig}

\usepackage{parskip}
 
\usepackage[amsthm]{ntheorem}
\usepackage{thmtools}
\usepackage{amsmath}
\usepackage{amssymb}
\usepackage{graphicx}
 
\usepackage{verbatim}
\usepackage{url}
\usepackage{amsfonts} 
\usepackage{graphicx}
\usepackage{color}

\usepackage{mathptmx}

\usepackage{epsfig}
\usepackage{caption}
\usepackage{subcaption}

\newcommand{\st}{\, | \,}
\newcommand{\be}{\begin{equation}}
\newcommand{\ee}{\end{equation}}
\newcommand{\R}{\mathbb R}
\newcommand{\C}{\mathbb C}
\newcommand{\diag}{\operatorname{diag}}

\newcommand{\nrank}{\operatorname{nrank}}
\newcommand{\rank}{\operatorname{rank}}
\newcommand{\ind}{\operatorname{index}}
\newcommand{\spanop}{\operatorname{span}}

\DeclareMathOperator{\vol}{vol}
  
\newtheorem{Theorem}{Theorem}

\newtheorem{Definition}{Definition}

\newtheorem{Proposition}[Theorem]{Proposition}
\newtheorem{Remark} {Remark}
\newtheorem{Example} {Example}

\begin{document}

\begin{frontmatter}
%%%%%%%%%%%%%%%%%%%%%%%%%%
\title{The~$k$-Compound   of a Difference-Algebraic System}
\thanks[footnoteinfo]{
This work   was supported in part by a research grant  from the~ISF.}
 \author[First]{Ron Ofir}
 \author[Second]{Michael Margaliot}
%%%%%%%%%%%
\address[First] {Ron Ofir is with  
		the Andrew and Erna Viterbi Faculty of Electrical and Computer Engineering,
		Technion---Israel Institute of Technology, Haifa 3200003, Israel.}
%%%%%%%%%%%%%%%%%%%%%%%%%%%%%%%%%%%%%%%%
\address[Second]{Michael Margaliot (Corresponding Author) is with the Department of  Elec. Eng.-Systems and the Sagol School of Neuroscience,  Tel-Aviv
 University, Tel-Aviv 69978, Israel. E-mail:  \texttt{michaelm@tauex.tau.ac.il}}
%%%%%%%%%%%%%%%%%%%%%%%%%

\maketitle

\begin{abstract}
The multiplicative and additive compounds of a matrix have important applications in geometry, linear algebra, and the analysis of dynamical systems.
In particular, the~$k$-compounds  allow to build a $k$-compound dynamical system that tracks the evolution of $k$-dimensional parallelotopes along the original dynamics. 
This has recently found many applications in the analysis of non-linear systems described by~ODEs and difference equations. 
Here, we introduce the $k$-compound system 
corresponding to a differential-algebraic system, and describe 
several applications to the analysis of discrete-time dynamical systems described by 
difference-algebraic equations. 
\end{abstract}

%%%%%%%%%%%%%%%

\end{frontmatter}

\section{Introduction}
%%%%%%%%%%%%%%%%%%%%%%%%%%%%%%%%

There is a growing interest in the applications of compound matrices to systems and control theory~(see, e.g.~\cite{margaliot2019revisiting,eyal_omri_tutorial,rami_osci,rola_spect,CTPDS,  pines2021, DT_K_POSI, gruss1,gruss2,gruss3,gruss4,gruss5}). In particular,   $k$-compound matrices   have   been used to generalize important classes  of dynamical systems  leading to linear~$k$-positive and  non-linear $k$-cooperative   systems~\citep{Eyal_k_posi,9107214},
$k$-contracting systems~\citep{muldowney1990compound,kordercont,ofir2021sufficient}, and~$k$-diagonally stable systems~\citep{cheng_diag_stab}.

Here, we introduce the $k$-compound system of
the  difference-algebraic  equation~(DAE):
%%%%%%%%%%%%%%%%
\begin{equation}\label{eq:DAE_TINV}
    B x(j+1) = A x(j),\quad j=0,1,\dots, 
\end{equation}
and also the~$k$-compound  of the corresponding 
matrix pencil.

Given matrices~$A,B\in \mathbb{C} ^{n\times n}$, the corresponding matrix pencil is the matrix polynomial 
\be\label{Eq:mtpen} 
(A,B) := A - \lambda B,\quad \lambda\in\C.
\ee
Matrix  pencils have numerous applications in linear algebra and in systems and control theory~(see e.g., the monographs~\citep{kunkel2006dae,milano2021eigenvalue,DAES2015}).
 %%%  
In particular, matrix pencils and their generalized  eigenvalues and eigenvectors
play an important role in the analysis of DAEs
(see~\citep{kunkel2006dae,gene_inverse}).

We show that the $k$-compound system of a~DAE describes the evolution of~$k$-parallelotopes under the~DAE. We   analyze properties such as consistency of initial  conditions, tractability, and stability of the $k$-compound system and relate them to the analogue properties in the original~DAE. We also present a result that relates the stability of the $k$-compound system to the existence of a stable subspace of the original~DAE. The dimension of this stable subspace 
decreases as~$k$ increases.

Several papers considered matrix pencils and used 
matrix compounds in their analysis~\citep{karcanias1994matrix,mitrouli1996numerical,kalogeropoulosnecessary,karcanias1989geometric,iwata2003}.
This is 
closely related to the Smith form of a matrix polynomial (see, e.g.,~\cite{mat_poly}), and also to the GCD of a given set of polynomials (see, e.g.~\cite{greek1989}).
However, to the
best of our knowledge the $k$-compound system of an DAE
that we introduce here, and its applications, are novel.

The remainder of this note  is organized as follows. 
The next section briefly reviews some known results that are used later on. 
Section~\ref{sec:kDAE} introduces 
the $k$-multiplicative compound system associated with a DAE. This is also a DAE, and its analysis, described in Section~\ref{sec:analysiskDAE}, 
is based on the so called   
$k$-multiplicative compound matrix pencil introduced in Section~\ref{sec:kpencil}. Section~\ref{sec:appli}
shows an application of these theoretical results by linking the stability of the~$k$-multiplicative compound system with that of the original~DAE. Our results use the 
Drazin inverse of the $k$-multiplicative compound of a matrix. For the sake of completeness, it is shown in the Appendix that this is equal to the $k$-multiplicative of the Drazin inverse of the original matrix.

We use standard  notation. Vectors [matrices] are denoted by small [capital] letters. 
A square matrix~$A$ is called regular [singular] if~$\det(A)\not =0$ [$\det(A)=0$].
The   complex conjugate transpose of~$A$ is denoted by~$A^*$.
If~$A$ is real then this is just the transpose of~$A$ denoted~$A^T$.
%%%%%%%%%%%%%%%%%%
Given an integer~$ n\geq 1$  and~$k\in\{1,\dots,n\}$, let~$Q(k,n)$ denote  the list of  all $k$-tuples:~$\alpha_1<\dots<\alpha_k$, with~$\alpha_i \in\{1,\dots,n\}$, ordered lexicographically. For example, 
\be\label{eq:q34}
Q(3,4) = \left ( (1,2,3),(1,2,4),(1,3,4),(2,3,4) \right ).
\ee
We refer to the $i$th element of~$Q(k,n)$ as~$\alpha^i$ and use subscripts to refer to an entry, e.g. if $k=3$ and $n=4$, then $\alpha_1^4 = 2$.
Given $\alpha,\beta \in Q(k,n)$, let $A[\alpha|\beta]$ denote the submatrix of $A$ obtained by taking  the rows [columns] with indices in $\alpha$ [$\beta$]. For example, 
$
A[(1,2)|(1,3)]=\begin{bmatrix}
    a_{11} & a_{13} \\
    a_{21} & a_{23}
\end{bmatrix}.
$
Let~$A(\alpha|\beta) := \det(A[\alpha|\beta])$ denote the corresponding minor.
%%%%%%%%%%%%%%%%%
For square matrices~$A_i \in \C^{n_i\times n_i}$, $i=1,\dots,\ell$,
let~$\diag_{i\in \{1,\dots,\ell\}}(A_i)$    denote the~$(\sum_{i=1}^\ell n_i )\times( \sum_{i=1}^\ell n_i)$ diagonal block matrix with blocks~$A_1,  \dots,A_\ell$.
The~$n\times n$ identity matrix is denoted by~$I_n$. 

%%%%%%%%%%%%%%%%%%%%%%%
\section{Preliminaries}
%%%%%%%%%%%%%%%%%%%%%%%%%%%%%
To make this paper more self-contained, we first 
review  the multiplicative   compound
of a matrix (for more details see, e.g.,~\cite{muldowney1990compound,fiedler_book}), and    provide   a very brief overview of 
matrix pencils and their applications in~DAEs. 

\subsection{Multiplicative compound of a matrix}
%%%%%%%%%%%%%%
Let~$A\in\mathbb{C}^{n \times m}$, and fix~$k\in\{1,\dots,\min\{n,m\}\}$. A $k$-minor of~$A$ is the determinant of a $k \times k$ submatrix of~$A$. The \emph{$k$-multiplicative compound} of~$A$, denoted~$A^{(k)}$, is   the~$\binom{n}{k}\times \binom{m}{k}$ matrix such that
$
    (A^{(k)})_{ij} = A(\alpha|\beta),
$
where~$\alpha$~[$\beta$] is the $i$th~[$j$th] sequence in~$Q(k,n)$.
For example, if $A \in \C^{3 \times 3}$ then  
\[A^{(2)}=
\begin{bmatrix}
A((1,2)| (1,2) ) &  A((1,2)| (1,3) ) & A((1,2)| (2,3) ) \\
A((1,3)| (1,2) ) &  A((1,3)| (1,3) ) & A((1,3)| (2,3) ) \\
A((2,3)| (1,2) ) &  A((2,3)| (1,3) ) & A((2,3)| (2,3) )
\end{bmatrix}.
\]

In particular,~$A^{(1)} = A$, and if~$m=n$ then~$A^{(n)} = \det(A)$.

The Cauchy-Binet theorem asserts that for any~$A\in\C^{n\times m}$, $B\in \C^{m\times \ell}$, and~$k\in\{1,\dots, \min\{n,m,\ell)\}$, we have
\be\label{eq:cb}
    (AB)^{(k)}=A^{(k)} B^{(k)}.
\ee
This justifies the term multiplicative compound.
When~$k=n=m=\ell$, \eqref{eq:cb} reduces to~$\det(AB)=\det(A)\det(B)$. 

 The definition of the multiplicative compound implies that~$(A^*)^{(k)}= (A^{(k)})^*$, and~$I_n^{(k)}=I_r$, with~$r:=\binom{n}{k}$. If~$A$ is square and regular  then~\eqref{eq:cb} gives
\begin{align*}
    I_n^{(k)} &= (A^{-1}A) ^{(k)} =(A A^{-1}) ^{(k)} \\
              &= (A^{-1})^{(k)} A ^{(k)} =A ^{(k)} (A^{-1})^{(k)} ,
\end{align*}
so~$A^{(k)}$ is also regular  and its inverse is~$(A^{-1})^{(k)}$. 
A similar argument shows that if~$U\in\C^{n\times n}$ is unitary, that is,~$U^*U= UU^*=I_n$,   then~$(U^{(k)})^* U^{(k)}= U^{(k)} (U^{(k)})^*=I_r$, so~$U^{(k)}$ is also unitary.

If $A$ is upper triangular (lower triangular, diagonal) 
then~$A^{(k)}$ is upper triangular (lower triangular, diagonal), and the diagonal entries of~$A^{(k)}$ are 
\begin{equation} \label{eq:akii}
    (A^{(k)})_{i,i} = \prod_{j=1}^k a_{{\alpha^i_j}, {\alpha^i_j}}
\end{equation}
where $\alpha^i$ is the $i$th sequence in $Q(k,n)$.
For example, for~$n=4$ and~$k=3$, 
$
\alpha^2=(1,2,4)$,  so~\eqref{eq:akii} becomes 
$
(A^{(3)})_{2,2} =  a_{11}a_{22}a_{44}.
$

If~$\lambda_i$, $i=1,\dots,n$, are the eigenvalues of~$A\in\mathbb{C}^{n \times n}$
then the eigenvalues of~$A^{(k)}$ are the $\binom{n}{k}$ products:
$
\prod_{i=1}^k \lambda_{\alpha_i}$, $ \alpha \in Q(k,n).
$
For $k=n$ this reduces to~$\det(A) = \prod_{i=1}^n \lambda_i$.

If~$A$ is rectangular, the singular values of~$A$ and
of~$A^{(k)}$ satisfy the same multiplicative property as the eigenvalues. Then using the singular value decomposition and the Cauchy-Binet theorem  yields 
%%%%%%%%%
\begin{equation}\label{eq:k_mul_rank}
    \rank(A^{(k)}) = \binom{\ell}{k},
\end{equation}
where~$\ell:=\rank(A)$, and~$\binom{\ell}{k}$ is defined as zero when~$k>\ell$. It follows from~\eqref{eq:k_mul_rank} 
that~$A^{(k)}=0$ if and only if~(iff)~$k>\rank(A)$, and   that $A^{(k)}$ has non-full rank iff~$A$ has non-full rank.

One reason  for the usefulness of the $k$-compounds in systems and control theory 
is that they have an important geometric application.
%%%%%%%%%%%%%%%%%%%%%%%%%%%%%%%%
\subsubsection{Geometric interpretation of the  multiplicative compound}
%%%%%%%%%%%%%%%%
Fix~$k$ vectors~$x^1,\dots,x^k \in \R^n$, and let~$P(x^1,\dots,x^k) := \{ \sum_{i=1}^k s_i x^i \st s_j \in[0,1] \}$ denote the parallelotope with  vertices~$x^1,\dots,x^k$ and~$0$ (see Fig.~\ref{fig:parallelogram}). Let~$\text{vol}(P)$ denote the volume of~$P$. Define the~$n\times k$ matrix $X:=\begin{bmatrix}
    x^1&\dots&x^k
\end{bmatrix}$,
and the~$k \times k$ non-negative definite matrix 
$
    G(x^1,\dots,x^k)  :=  X^T X$.
%%%
Then~$\vol(P) = \sqrt{\det(G )}$  \cite[Ch.~IX]{Gantmacher_vol1}. To express this using the multiplicative compound, note that~$\det(G)=G^{(k)}$, so
 %%%%%%%%%%%%%% 
$
    \vol(P)  =\sqrt{ (X^T X)^{(k)} } 
     = \sqrt{ (X^{(k)} ) ^T X ^{(k)} }.
$
By definition, the dimensions of~$X^{(k)}$ are~$\binom{n}{k}\times \binom{k}{k}$, i.e. it is an~$\binom{n}{k}$-dimensional  column vector, so we conclude that
\begin{equation}\label{eq:volpeqn}
\vol(P)=|X^{(k)}|_2,
\end{equation}
where $|\cdot|_2$ denotes the $L_2$ norm. In the special case~$k=n$, this reduces to  the well-known formula 
$$
\vol(P(x^1,\dots,x^n))=|\det(\begin{bmatrix}
    x^1&\dots & x^n
\end{bmatrix})|.
$$

\begin{figure}
    \centering
    \begin{tikzpicture}
        \draw[dashed,pattern=dots] (0,0,0)--(3,0.5,0)--(4,2.5,0)--(1,2,0)--cycle;
        \draw[thick,->] (0,0,0)--(3,0.5,0) node[right]{$x^1$};
        \draw[thick,->] (0,0,0)--(1,2,0) node[above]{$x^2$};
        \draw (0,0,0) node[below]{0};
        \draw (2,1.25,0) node[fill=white,inner sep=0.5pt]{$P(x^1,x^2)$};
    \end{tikzpicture}
    \caption{  2D parallelotope with  vertices~$0$, $x^1$, and~$x^2$.}
    \label{fig:parallelogram}
\end{figure}
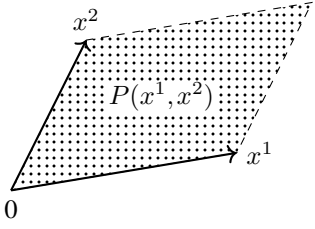

\subsection{Matrix pencils}
%%%%
Given~$A,B \in \C^{n \times n}$, the  associated matrix pencil is the matrix polynomial~\eqref{Eq:mtpen}. 
%%%
The matrix pencil is called \emph{regular} if there exists a~$\lambda\in \C$ such that~$\det(A - \lambda B) \not = 0$. Otherwise, it is called singular.
%%%%%%%%
The normal rank of~$(A,B)$ is
\[
    \nrank(A,B) := \max_{\lambda \in \C} \rank(A - \lambda B).
\] 
Any $\lambda_0 \in \C$  for which 
\begin{equation}
    \rank(A - \lambda_0 B) < \nrank(A,B) 
\end{equation}
is called a finite (generalized) eigenvalue of~$(A,B)$.
For any such~$\lambda_0$ there exists a vector~$v \in \C^n\setminus\{0\}$ such that
\begin{equation}\label{eq:fgeigen}
    A v = \lambda_0 B v.
\end{equation}
If~$(A,B)$ is regular then~\eqref{eq:fgeigen} implies that~$\lambda_0$ is a finite eigenvalue of~$(A,B)$.
%%%%
If~$\det(B) = 0$ then~$(A,B)$ also has an eigenvalue at infinity, which corresponds to the zero eigenvalue of the matrix pencil~$(B,A)$.

Any~$A,B \in \C^{n \times n}$  may be jointly  triangularized using the generalized Schur decomposition~(GSD)  \cite[Thm. 7.7.1]{golub2013matrix},
%%%%%%%%%%%%%%%%%%%%%%%
that is, there exist unitary matrices~$U,V\in\C^{n\times n}$ such that
\begin{equation}\label{eq:jointtrig}
    U A V = T, \quad
    U B V = S,
\end{equation}
where~$T$ and~$S$ are upper triangular. The GSD is particularly useful when studying the spectrum of a matrix pencil, as
\begin{align*}
    \det(A - \lambda B) &= \det(U)\det(T - \lambda S)\det(V) \\
    &= \det(U)\det(V) \prod_{i=1}^n (T_{ii} - \lambda S_{ii}).
\end{align*}
Thus,~$(A,B)$ is singular iff there exists~$i \in \{1,\dots,n\}$ such that~$T_{ii} = S_{ii} = 0$.
If~$(A,B)$ is regular, then its eigenvalues may be read from the diagonal entries of~$T$ and~$S$: for every $i \in \{1,\dots,n\}$ such that $S_{ii} \neq 0$, $T_{ii} / S_{ii}$ is a finite eigenvalue of~$(A,B)$, and~$(A,B)$ has an eigenvalue at infinity iff there exists $i \in \{1, \dots,n\}$ such that $S_{ii} = 0$.

\subsection{Difference-algebraic equations}
%%%%
Consider the DAE~\eqref{eq:DAE_TINV}  
with~$x:\{0,1,\dots\} \to \R^n$, and~$B,A\in\R^{n\times n}$.
If~$B$ is regular then this is equivalent to the discrete-time LTI  system~$x(j+1)=B^{-1}Ax(j)$, 
but we will assume that~$B$ is singular. 
Then~\eqref{eq:DAE_TINV} 
may be interpreted as a discrete-time dynamical system with algebraic constraints.

An initial condition~$x(0)$ is called \emph{consistent} if~\eqref{eq:DAE_TINV}
admits a corresponding solution~$x(j)$    for all~$j\geq 0$. For example,~$x(0)=0$ is always   consistent. 
The system~\eqref{eq:DAE_TINV} is called \emph{tractable} (some authors use instead the term \emph{solvable}) if for any consistent initial condition~$x(0)$
the system~\eqref{eq:DAE_TINV}  admits a \emph{unique} solution~$x(j)$, $j=0,1,\dots$.

The next two  results relate  the system-theoretic properties of~\eqref{eq:DAE_TINV} to  the matrix pencil~$(A,B)$. To state them, we recall 
the notions of the Drazin index and the Drazin inverse. 
\begin{Definition}\citep{drazin1958}
The~\emph{Drazin index}
of a square matrix $A$,
denoted $\ind(A)$,
is the minimal integer~$k\ge 0$ such that 
$
    \rank (A^{k}) = 
    \rank ( A^{k+1} ).
 $
%%%%%%%%%%%%%%%
\end{Definition}
%%%%%%%%%%%%%%%
For example, if~$A$ is regular then~$\rank(A^0)=\rank(A^1)$, so~$\ind(A)=0$. 
If~$N$ is nilpotent, i.e. there exists a minimal integer~$k$ such that~$N^k=0$,  then~$\ind(N)=k$.
%%%%%%%%%%%%%%%%%%%%%%%%%%%%%%%%%%
\begin{Definition}\label{def:draz_inv}\citep{drazin1958}
%%%%%%%%%%%%%%%
%%%%%%%%%%%%%%%
%
Let $A$ be a square matrix.
The~\emph{Drazin inverse}
of $A$ is a matrix $X$ such that 
\begin{enumerate}
    \item $A^{\ind (A) +1} X 
    = A^{\ind (A)}$,
    \item $AX = XA$,
    \item $XAX = X$.
\end{enumerate}
%
%
%%%%%%%%%%%%%%%
%%%%%%%%%%%%%%%
\end{Definition}
%%%%%%%%%%%%%%%
%%%%%%%%%%%%%%%
%
It is known that the Drazin inverse, denoted~$A^D$, always exists and is unique.  
If~$A$ is regular then~$A^D=A^{-1}$, and  
if~$N$ is nilpotent then~$N^D=0$. 
If
the Jordan decomposition of~$A$ is
\be\label{eq:jordan_A}
A = T^{-1}
\begin{bmatrix}
C & 0 \\
0 & N
\end{bmatrix}
T,
\ee
with~$C$ regular and~$N$ nilpotent, 
then
$
A^{D} = T^{-1}
\begin{bmatrix}
C^{-1} & 0 \\
0 & 0
\end{bmatrix}
T.  
$
%%%%%%%%%%%%%%%%%%%%%%%%%%%%%%%%%%%%%%%%%%
\begin{Proposition}
\label{prop:tract}\citep{gene_inverse}
%%%%%%%%%%%%%%%
%%%%%%%%%%%%%%%
%
The DAE~\eqref{eq:DAE_TINV}
is~{tractable} iff there exists~$\lambda\in\C$
such that
$\det (A- \lambda B   ) \ne 0$, that is, iff~$(A,B)$ is regular.
%%%%%%%%%%%%%%%%%%%%%%%%%%
\end{Proposition}

\begin{Proposition}\citep{gene_inverse,Belov2018}
\label{prop:x0_consist}
%%%%%%%%%%%%%%%
%%%%%%%%%%%%%%%
%
Assume that~\eqref{eq:DAE_TINV}
is tractable.
Fix $\lambda\in\mathbb{C}$  such that 
$\det (A- \lambda B   ) \ne 0$,
and let
\begin{align}\label{eq:change1}
    \hat{B}_{\lambda}:= (A-\lambda B  )^{-1} B, \quad
    \hat{A}_{\lambda}:= (A-\lambda B  )^{-1} A.
\end{align}
Let $i:=\ind (\hat{B}_{\lambda})$.
An initial condition $x(0)$ is~\emph{consistent}
iff $
x(0)$ is in the range of~$   (\hat{B}_{\lambda})^{i}  $,
and for such an initial condition the unique 
solution of~\eqref{eq:DAE_TINV} is
\begin{align}\label{eq:chnage2}
    x(j ) = 
    \left ( 
    ( \hat{B}_{\lambda} )^{D}
    \hat{A}_{\lambda}
    \right )^{j}
    x(0), \quad 
    j=0,1, \dots.
\end{align}
Furthermore,~$\lim_{j\to \infty} x(j)=0$ for any consistent initial condition~$x(0)$ iff
 all the finite eigenvalues of~$(A,B)$ lie in the open unit disk.
%%%%%%%%%%%%%%%
%%%%%%%%%%%%%%%
\end{Proposition}

The next sections describe our main results.
%%%%%%%%%%%%%%%%%%%%%%%

\section{The $k$-multiplicative compound DAE}\label{sec:kDAE}
%%%%%%%%%%
We begin by generalizing   a DAE  to
a corresponding  $k$-compound DAE. 
We consider the general case of a
    time-varying  DAE: 
%%%%%%%
\begin{equation}\label{eq:lin_dae_dt}
    B(j+1) x(j+1) = A(j)x(j),\quad j=0,1,\dots,
\end{equation}
 with~$x\in\R^{n}$, and~$A,B\in\R^{n\times  n }$.

We introduce a new definition.
 %%%%%%%%%%%%%%%%
\begin{Definition}
%%%%%%%%%%%%%%%%%%%
Fix~$k\in\{1,\dots,n\}$, and let~$r:=\binom{n}{k}$. 
 The~$k$-multiplicative 
  compound DAE corresponding to~\eqref{eq:lin_dae_dt} is
\begin{equation}\label{eq:compound_dae_dt}
    B^{(k)}(j+1) y(j+1) = A^{(k)}(j) y(j), \quad j=0,1,\dots,
\end{equation}
with~$y\in\R^r$.
%%%%%%%%%%%%%%%%%%%%%%%%%%%%
\end{Definition}

Note that for~$k=1$, Eq.~\eqref{eq:compound_dae_dt} is the original DAE~\eqref{eq:lin_dae_dt}, whereas for~$k=n$,  Eq.~\eqref{eq:compound_dae_dt} becomes the scalar equation: 
$
\det(B(j+1))
y(j+1)
 =
 \det(A(j))y(j) .
$

The next result shows that the $k$-compound DAE  
tracks  the evolution of  volumes of $k$-parallelotopes
  under the DAE~\eqref{eq:lin_dae_dt}.
  %%%%%%%%%%%%%%%%%%%%%%%%%%%%%%%
\begin{Proposition}\label{prop:compound_dae_dt}
%%%%%%%%%%%%%%%%%%%%%%%%%%%%%%%%
Fix $k \in \{1,\dots,n\}$, and pick~$k$
consistent 
initial conditions~$a^1,\dots,a^k \in \mathbb{R}^n$ of~\eqref{eq:lin_dae_dt}. Let~$x^i(\ell) := x(\ell,a^i)$    denote   a    solution
at time~$\ell$ of~\eqref{eq:lin_dae_dt}   emanating from~$x(0)=a^i$. Define the $n \times k$ matrix 
\[
X(j) := \begin{bmatrix} x^1(j) & \dots & x^k(j) \end{bmatrix} , 
\]
and the~$\binom{n}{k}$-dimensional  column vector
\begin{equation}\label{eq:kcomp_y}
y(j):=(X(j) )^{(k)}. 
\end{equation}
Then~$y$ is a solution of  the $k$-compound DAE~\eqref{eq:compound_dae_dt}.
\end{Proposition}
%%%
%%%
The proof is straightforward. By~\eqref{eq:lin_dae_dt},
$    B(j+1) X(j+1) = A(j) X(j).
$
Taking the $k$th multiplicative compound on both sides, and using the Cauchy-Binet Theorem completes the proof.

%%%%%%%%%%%%%%%%%%%%%%%%%%% 
\section{The $k$-multiplicative compound matrix pencil}\label{sec:kpencil}
 %%%%%%%%%%%%%%%%%%
  In the time-invariant case,  the $k$-compound DAE is 
\be\label{eq:indu_ti}
B^{(k)} y(j+1) = A^{(k)}  y(j).
\ee
Our next goal is to extend known analysis results for time-invariant DAEs to~\eqref{eq:indu_ti}. Since the known analysis results are closely related to the pencil~$(A,B)$, we begin by introducing the~$k$-pencil associated to the~$k$-multiplicative compound  DAE.  
 
 \begin{Definition}
    Given~$A,B \in \C^{n \times n}$, and~$k\in\{1,\dots,n\}$,
   the $k$-multiplicative compound  pencil of~$(A,B)$ is the matrix pencil 
    \begin{equation}\label{eq:our_mat_pencil}
        (A,B)^{(k)} := A^{(k)} - \lambda B^{(k)},\quad
        \lambda\in\C.
    \end{equation}
    %%%%%%
 \end{Definition}
 %%%%%%%%%%%%%%%%
Note that~$(A,B)^{(1)}$ is just~$(A,B)$,
and~$(A,B)^{(n)} = \det(A) - \lambda \det(B)$. Also,~$(A,0)^{(k)}$ is just~$A^{(k)}$. 
%%%%%%%%%%%%%%
 \begin{Remark}
 Several authors associate with~$(A,B)$ the matrix polynomial
 \be\label{eq:matrix_poly}
 (A-\lambda B)^{(k)},\quad \lambda\in\C.
 \ee 
 This algebraic construction is particularly useful in studying the Smith normal form of~$(A,B)$.
 It should be 
 noted that~\eqref {eq:our_mat_pencil}
 and~\eqref{eq:matrix_poly} are quite different.  For example, let~$n=2$, $B=I_2$, and
$
    A = \diag(
        \mu_1 , \mu_2)
$. 
Then
\[    (A - \lambda B)^{(2)} = (\mu_1 - \lambda)(\mu_2 - \lambda) , 
\] 
whereas \eqref{eq:our_mat_pencil}
gives
\[
(A,B)^{(2)} =    A^{(2)} - \lambda B^{(2)} = \mu_1 \mu_2 - \lambda.
\] 
In particular, here  the generalized eigenvalues of $(A - \lambda B)^{(2)}$ are simply the generalized eigenvalues of $(A - \lambda B)$, whereas $(A,B)^{(2)}$  admits a  single generalized eigenvalue~$\mu_1 \mu_2$. In this sense, neither $(A,B)^{(k)}$ nor~$(A - \lambda B)^{(k)}$ is a generalization or a special case of the other.
 %%%%%%%%
%%%%%%% 
 %%%%%%
 \end{Remark}

Propositions~\ref{prop:tract} and~\ref{prop:x0_consist} imply that the matrix pencil $(A,B)$ determines important system-theoretic properties of~\eqref{eq:lin_dae_dt}, and thus the matrix pencil $(A,B)^{(k)}$ determines the same system-theoretic properties for~\eqref{eq:indu_ti}.
In particular, it follows from Prop.~\ref{prop:tract} that the $k$-compound system \eqref{eq:indu_ti} is tractable iff the matrix pencil~$(A,B)^{(k)}$ is regular.
 
\subsection{Regularity of~$(A,B)^{(k)}$}
%%%%%%%%%%%%%%%%%%%%%%%%%%%
The next result provides a necessary and sufficient condition for regularity of the pencil~$(A,B)^{(k)}$, with~$k\geq 2$.

\begin{Proposition}
    \label{prop:mul_regular}
%%%
Let $A,B \in \C^{n \times n}$.  Fix~$k \in \{2,\dots,n\}$.  The following  four   conditions are equivalent. 
    \begin{enumerate}
        \item \label{cond:pencil_singular} The pencil~$(A,B)^{(k)}$ is singular. 
        \item \label{cond:triu_singular}
        For any GSD~\eqref{eq:jointtrig}
there exists~$\alpha\in Q(k,n)$ such that~$\prod_{i=1}^k T_{\alpha_i,\alpha_i}=\prod_{i=1}^k S_{\alpha_i,\alpha_i}=0.
$
%%%%%%
\item \label{cond:A_B_singular}
        $\det(A)=\det(B)=0$.
        \item \label{cond:kernels_intersect}
        $\ker(A^{(k)}) \cap \ker(B^{(k)})\not = \{0\}.$
%%%%%%
    \end{enumerate}
\end{Proposition}

\begin{proof}
%%%%
We begin by showing that~\eqref{cond:pencil_singular} and~\eqref{cond:triu_singular} are equivalent. Applying the Cauchy-Binet Theorem to~\eqref{eq:jointtrig}
gives
    \begin{equation}\label{eq:triang_k}
        U^{(k)}A^{(k)}V^{(k)} = T^{(k)}, \quad
        U^{(k)}B^{(k)}V^{(k)} = S^{(k)},
    \end{equation}
    and~$T^{(k)},S^{(k)}$ are also upper triangular.
Thus,
\begin{align}\label{eq:detakbk}
    \det( (A,B)^{(k)} ) &= \det ( (U^{(k)})^* (V^{(k)})^*    )  \det(T^{(k)}-\lambda S^{(k)}) , 
\end{align}
so~$
\det(A^{(k)}-\lambda B^{(k)})=0
$
iff
$
\prod_{i=1}^r 
  \left  ((T^{(k)})_{ii} -\lambda (S^{(k)})_{ii}\right) =0,
$
where~$r:=\binom{n}{k}$. 
%%%%%%%%
    In particular, $\det(A^{(k)}-\lambda B^{(k)})=0$  for any~$\lambda\in\C$
    iff
 there exists an~$i\in\{1,\dots,r\}$ such that~$(T^{(k)})_{ii} = (S^{(k)})_{ii}=0$.
%%%
Since~$T$ is upper triangular,  entry~$(i,i)$ of~$T^{(k)}$ is 
$\prod_{\ell=1}^k T_{\alpha^{i}_\ell,\alpha^{i}_\ell}$,
where~$\alpha^{i}$ is the~$i$th sequence in~$Q(k,n)$, and similarly for~$S^{(k)}$. Thus,
  \eqref{cond:pencil_singular} and~\eqref{cond:triu_singular} are equivalent.
%%%

Suppose that~$(A,B)^{(k)}$
is singular. Then at least one diagonal entry of~$T$  and at least one diagonal entry of~$S$  are zero, so
\begin{align*}
    \det(A)&=\det(U)\det(V)\det(T)=0,\\
    \det(B)&=\det(U)\det(V)\det(S)=0.
\end{align*} 
This proves that~\eqref{cond:pencil_singular}
implies~\eqref{cond:A_B_singular}. 

We now show that~\eqref{cond:A_B_singular}
implies~\eqref{cond:kernels_intersect} for~$k=2$.
Assume that~\eqref{cond:A_B_singular} holds. Then there exist~$x,y\in\C^n\setminus\{0\}$ such that~$Ax=0$, $By=0$.
We consider two cases. If~$x,y$ are linearly dependent then~$x=s y$ for some 
scalar~$s\not =0$.
Pick~$z\in\R^n$ such that~$x,z$ are linearly independent. Then
\begin{align}
    A^{(2)}\begin{bmatrix} x& z \end{bmatrix}^{(2)} &= 
    \begin{bmatrix} Ax& Az \end{bmatrix}^{(2)}=
    \begin{bmatrix} 0& Az \end{bmatrix}^{(2)}=0,
\end{align}
and similarly $B^{(2)}\begin{bmatrix} x& z \end{bmatrix}^{(2)}=B^{(2)}\begin{bmatrix}  s y& z \end{bmatrix}^{(2)}=0$, so~\eqref{cond:kernels_intersect} holds.
Now assume that~$x,y$ are linearly independent. Let~$z:=\begin{bmatrix}
  x& y
  \end{bmatrix}^{(2)}$. Then 
\begin{align*}
  A^{(2)} z=\begin{bmatrix}
  Ax& Ay
  \end{bmatrix}  ^{(2)} 
  =\begin{bmatrix}
  0& Ay
  \end{bmatrix}  ^{(2)}
  =0,
\end{align*}
and similarly~$B^{(2)}  z=0$, so
  $z\in\ker(A^{(2)} ) \cap\ker(B^{(2)})$. 
Since~$x,y$ are linearly independent,~$z\not =0$. 
  A similar argument shows that for any~$j\in\{2,\dots, n\}$, we have~$\ker(A^{(j)} ) \cap\ker(B^{(j)}) \not =\{0\}$, 
  so~\eqref{cond:A_B_singular}
implies~\eqref{cond:kernels_intersect}. 
 
 Now suppose that~\eqref{cond:kernels_intersect} holds. Let~$x\not =0$ be a vector such that~$x\in \ker(A^{(k)} )\cap \ker(B^{(k)} )$. Then
$
 ( A^{(k)}  -\lambda B ^{(k)})x=0 
 $
 for any~$\lambda \in \C$, so~\eqref{cond:pencil_singular} holds.
 We conclude that~\eqref{cond:kernels_intersect}  implies~\eqref{cond:pencil_singular},  
  and this completes the proof of Prop.~\ref{prop:mul_regular}.
%%%%%
\end{proof}

\begin{Remark}
    Prop.~\ref{prop:mul_regular} implies in particular that if~$(A,B)$ is singular then~$(A,B)^{(k)}$ is singular for any~$k\in\{1,\dots,n\}$.  Furthermore, either~$(A,B)^{(k)}$ is regular for all $k>1$, or it is singular for all $k>1$.
\end{Remark}

Prop.~\ref{prop:mul_regular} demonstrates a perhaps surprising property of the $k$-multiplicative compound of a pencil. A sufficient, but not necessary, condition for a pencil~$(A,B)$ to be singular is that~$\ker(A) \cap \ker(B) \neq \{0\}$. However, this condition is both  sufficient and \emph{necessary} for the singularity of
for~$(A,B)^{(k)}$, with~$k > 1$.

The next example illustrates Prop.~\ref{prop:mul_regular}.% and Corollary~\ref{coro:kmore1}.
%%%%%%%%%
\begin{Example}
%%%%%
Suppose that~$A=\diag(0,1,2)$ and~$B=\diag(1,2,0)$. Note that $\det(A)=\det(B)=0$.
Then~$\det(A-\lambda B) =2\lambda (2\lambda-1)$, so~$(A,B)$ is regular.
Also,~$A^{(2)}=\diag(0,0,2) $ and~$B^{(2)} = \diag(2,0,0)$, so
$\det(A^{(2)}-\lambda B^{(2)}) =0$ for any~$\lambda\in \C$ and thus~$(A,B)^{(2)}$ is singular. Consider the 
GSD  in~\eqref{eq:jointtrig}.
Then there exists exactly one~$i\in\{1,2,3\}$ such that~$T_{ii}=0$, and  exactly one~$j\in\{1,2,3\}$ such that~$S_{jj}=0$. Also,~$i\not =j$, as otherwise~$(T,S)$ is singular and this is impossible as~$(A,B)$ is regular. Let~$\alpha$ be the sequence in~$Q(2,3)$ 
that includes~$i$ and~$j$. Then 
$
\prod_{r=1}^2 T_{\alpha_r,\alpha_r}=\prod_{i=1}^2 S_{\alpha_r,\alpha_r}=0.
$
Note also that~$  \begin{bmatrix} 0&1&0 \end{bmatrix}^T \in \ker(A^{(2)}) \cap \ker(B^{(2)})$.
%%%%%
\end{Example}

%%%%%%%%%%%%%%
\subsection{Spectral properties of~$(A,B)^{(k)}$}
%%%%%%%%%%%%%%%%%
Recall that if~$A\in\C^{n\times n} $ and~$k\in\{1,\dots,n\}$ then any eigenvalue of~$A^{(k)}$ is the product of~$k$ eigenvalues of~$A$. The next result demonstrates that  the matrix pencil~$(A,B)^{(k)}$ satisfies a similar property.

%%%%%%%%%%%%%%%%%%%%%%%
\begin{Proposition}\label{prop:mul_gen_eig}
    Let~$A,B \in \C^{n \times n}$. Fix~$k \in \{1,\dots,n\}$. Suppose that~$(A,B)^{(k)}$ is regular. Then every eigenvalue of~$(A,B)^{(k)}$ is the product of~$k$ eigenvalues of~$(A,B)$, where we define the product of infinity with any value as infinity.
\end{Proposition}
%%%%%%%%%%%%%%%%%

\begin{proof}
%%%
Let~$r:=\binom{n}{k}$. By Eq.~\eqref{eq:detakbk},
$
    \det((A,B)^{(k)}) = \det((T,S)^{(k)}) = \prod_{i=1}^r ((T^{(k)})_{ii}-\lambda (S^{(k)})_{ii}),
$
so the eigenvalues of~$(A,B)^{(k)}$ are
$
  { (T^{(k)})_{ii} }/{(S^{(k)})_{ii}}    
$, $ i\in\{1,\dots,r\}$, 
where we define~${c}/{0}$, with~$c\not=0$, as infinity (note that the assumption that the pencil is regular guarantees that the case~$(T^{(k)})_{ii} =( S^{(k)})_{ii}=0$ is not possible). In particular, the eigenvalues of~$(A,B)=(A,B)^{(1)}$ are 
$
  { T _{ii} }/{S_{ii}}, i\in\{1,\dots,n\}  .
$
Since~$T$ is upper triangular, Eq.~\eqref{eq:akii} implies that entry~$(i,i)$ of~$T^{(k)}$ is 
$\prod_{\ell=1}^r T_{\alpha^i_\ell,\alpha^i_\ell}$,
where~$\alpha^i$ is the~$i$th sequence in~$Q(k,n)$, and similarly for~$S^{(k)}$. This completes the proof. 
%%%%%%%%%%%%%%%%%%%
\end{proof}

Given~$k$ finite eigenvalues and the  corresponding~$k$ eigenvectors of~$(A,B)$, the following result gives an explicit formula for the corresponding eigenvalue and eigenvector of~$(A,B)^{(k)}$.
\begin{Proposition}
    Let $A,B \in \C^{n \times n}$, and pick $k \in \{1,\dots,n\}$. Suppose that $\lambda_1,\dots,\lambda_k \in \mathbb{C}$ and $v^1,\dots,v^k \in \mathbb{C}^n \setminus \{0\}$  satisfy 
    \begin{equation}\label{eq:gen_eig}
        A v^i = \lambda_i B v^i, \quad i=1,\dots,k.
    \end{equation}
    Define $\tilde{v} := \begin{bmatrix}v^1 & \dots & v^k\end{bmatrix}^{(k)}$ and $\tilde{\lambda} := \prod_{i=1}^k \lambda_i$.
    Then
    \begin{equation}
        A^{(k)} \tilde{v} = \tilde{\lambda} B^{(k)} \tilde{v}.
    \end{equation} 
 %%%%%%%%%%%%%
\end{Proposition}
%%%%
This implies in particular that  if~$v^1,\dots,v^k$ are linearly independent eigenvectors of $(A,B)$ with corresponding eigenvalues~$\lambda_1,\dots,\lambda_k$, then~$\tilde{v} \in \R^{\binom{n}{k}}\setminus\{0\}$ is an eigenvector of~$(A,B)^{(k)}$ corresponding to the  eigenvalue~$\tilde \lambda$.

\begin{proof}
%%%%%%%%%%%%%%%%%%%%%%%%
    Applying the Cauchy-Binet Theorem yields 
    \begin{align}\label{eq:pcn}
        A^{(k)} \begin{bmatrix} v^1 & \dots & v^k\end{bmatrix}^{(k)} &= \begin{bmatrix} A v^1 & \dots & A v^k \end{bmatrix}^{(k)} \nonumber\\
        &= \begin{bmatrix} \lambda_1 B v^1 & \dots & \lambda_k B v^k \end{bmatrix}^{(k)}\nonumber \\
        &= \tilde{\lambda} \begin{bmatrix} B v^1 & \dots & B v^k \end{bmatrix}^{(k)}  \\
        &= \tilde{\lambda} B^{(k)} \begin{bmatrix} v^1 & \dots &v^k  \end{bmatrix}^{(k)},\nonumber 
     \end{align}
     and this completes the proof. 
%%%%%%%%%%%%%%%%%%%%%%%
\end{proof}

\begin{Remark}
%%%%%%%%%%%%    
The multiplicative compound of a matrix pencil has a geometric interpretation similar to that 
of the multiplicative compound of a matrix. Let~$r:=\binom{n}{k}$. 
The dimensions of~$\tilde v$ are~$r \times \binom{k}{k}$, i.e. it is an~$r$-dimensional  column vector, and~$|\tilde v|_2$ is the volume of the 
 parallelotope with vertices~$0,v^1,\dots,v^k$. Eq.~\eqref{eq:pcn} thus  implies that
the volume of the parallelotope
generated by~$0,Av^1,\dots,Av^k$  is equal
to~$| \prod_{i=1}^k \lambda_i  |$ times the volume 
of the parallelotope generated by~$0,Bv^1,\dots,B v^k$.
Indeed, this  follows from~\eqref{eq:gen_eig}.
%%%%%%%%%%%%
\end{Remark}

 %%%%%%%%%%%%%%%%%%%%%%%%%%%%%%%%%%%%%%%
\section{Analysis of the $k$-multiplicative compound~DAE}\label{sec:analysiskDAE}
%%%%%%%%%%%%%%%%%%%%%%%%%%%%%%%%%%%%%%%%%%%    
It is clear that the algebraic properties of the pencil~$(A,B)^{(k)}$ are closely related to the dynamical properties of~\eqref{eq:indu_ti}. The following example demonstrates this relation from the point of view of the spectral properties studied in Prop.~\ref{prop:mul_gen_eig}.

\begin{Example}\label{exa:periodic}
    Consider~\eqref{eq:DAE_TINV} with~$n=3$,
    \begin{equation*}
        A = \begin{bmatrix}
            0 & 1 & 0 \\
            -1 & 0 & 2 \\
            -1 & 0 & 1
        \end{bmatrix}, \text{ and } B = \begin{bmatrix}
            1 & 0 & 0 \\
            0 & 1 & 0 \\
            0 & 1 & 0
        \end{bmatrix}.
    \end{equation*}
    Here~$\det(A) \neq 0$, so~$A$ is regular. As  discussed 
    in Subsection~\ref{subsec:tract} below, this implies that the 2-multiplicative compound system is tractable. Furthermore, the eigenvalues of~$(A,B)$ are $\mi,-\mi,\infty$, so by Prop.~\ref{prop:mul_gen_eig} $(A,B)^{(2)}$ has eigenvalues $1,\infty,\infty$. Therefore, given any  two consistent initial conditions of the~DAE, the corresponding 
    solution of the 2-compound~DAE is constant in time. It follows from Prop.~\ref{prop:compound_dae_dt} and the geometric properties of multiplicative compounds that the corresponding solution of the 2-compound~DAE equals the area of the parallelotope defined by the solutions of the~DAE. We would therefore expect this parallelotope to have a constant area. This can be seen in Fig.~\ref{fig:periodic_sim}, that  depicts  two trajectories of the~DAE with initial conditions~$\begin{bmatrix}1&1\end{bmatrix}^T$ and~$\begin{bmatrix}1.5&0.75\end{bmatrix}^T$, and the triangles they define (with area which is half of the area of the corresponding parallelotopes). The corresponding solution of the  2-compound~DAE is $y(j) \equiv 0.75$, and this agrees with
    the fact that the area of the triangles is constant in time.
%%%%%%%%%%%%
\end{Example}

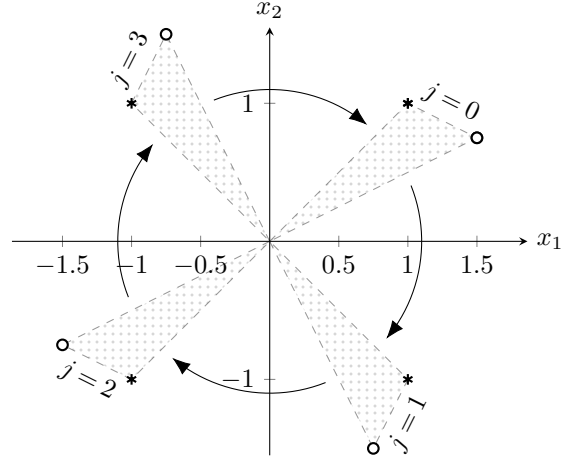
\begin{figure}
    \centering
    \begin{tikzpicture}
    \begin{axis}[
        xlabel = {$x_1$},
        ylabel = {$x_2$},
        yticklabel style = {
            /pgf/number format/fixed,
            /pgf/number format/precision=2
        },
        xmin = -1.55,
        xmax = 1.55,
        ymin = -1.55,
        ymax = 1.55,
        axis equal,
        axis lines=middle,
        every axis x label/.style={
            at={(ticklabel* cs:1)},
            anchor=west,
        },
        every axis y label/.style={
            at={(ticklabel* cs:1)},
            anchor=south,
        },
    ]
        \addplot[thick,only marks,mark=*,mark options={solid,fill=white}] table {
             1.5      0.75
             0.75    -1.5
            -1.5     -0.75
            -0.75     1.5
             1.5      0.75
        };
    
        \addplot[thick,only marks,mark=asterisk,mark options={solid,fill=white}] table {
             1     1
             1    -1
            -1    -1
            -1     1
             1     1
        };
    
        \draw [dashed, opacity=0.4, pattern=dots] (axis cs:{0,0}) -- (axis cs:{1,1}) -- (axis cs:{1.5,0.75}) -- cycle;
        \draw [dashed, opacity=0.4, pattern=dots] (axis cs:{0,0}) -- (axis cs:{0.75,-1.5}) -- (axis cs:{1,-1}) -- cycle;
        \draw [dashed, opacity=0.4, pattern=dots] (axis cs:{0,0}) -- (axis cs:{-1.5,-0.75}) -- (axis cs:{-1,-1}) -- cycle;
        \draw [dashed, opacity=0.4, pattern=dots] (axis cs:{0,0}) -- (axis cs:{-0.75,1.5}) -- (axis cs:{-1,1}) -- cycle;

        \path (axis cs:{1,1}) -- (axis cs:{1.5,0.75}) node[above, midway, sloped]{$j=0$};
        \path (axis cs:{1,-1}) -- (axis cs:{0.75,-1.5}) node[below, midway, sloped]{$j=1$};
        \path (axis cs:{-1,-1}) -- (axis cs:{-1.5,-0.75}) node[below, midway, sloped]{$j=2$};
        \path (axis cs:{-1,1}) -- (axis cs:{-0.75,1.5}) node[above, midway, sloped]{$j=3$};

        \tikzmath{\r = 1.1;};
        \addplot[-{Latex[length=7pt]},domain={atan(0.75/1.5)-5}:{-atan(1/1)+5}]({\r*cos(x)},{\r*sin(x)});
        \addplot[-{Latex[length=7pt]},domain={atan(0.75/1.5)-90-5}:{-90-atan(1/1)+5}]({\r*cos(x)},{\r*sin(x)});
        \addplot[-{Latex[length=7pt]},domain={180+atan(0.75/1.5)-5}:{180-atan(1/1)+5}]({\r*cos(x)},{\r*sin(x)});
        \addplot[-{Latex[length=7pt]},domain={90+atan(0.75/1.5)-5}:{atan(1/1)+5}]({\r*cos(x)},{\r*sin(x)});
    \end{axis}
    \end{tikzpicture}
%%%%%%
    \caption{Trajectories of the 3-dimensional DAE in Example~\ref{exa:periodic} projected onto the 2-dimensional subspace $\text{range}(\hat{B}_0)$. One trajectory is shown with asterisks, and the other with circles. The areas defined by the two trajectories are outlined with dashed borders.}
    \label{fig:periodic_sim}
\end{figure}
%%%%%%%%%%%%%%%%%%%

 %%%%%%%%%%%
\subsection{Tractability and asymptotic stability of the $k$-multiplicative compound DAE}\label{subsec:tract}
We begin by studying the tractabiliy of the $k$-multiplicative compound DAE~\eqref{eq:indu_ti}. Combining Prop.~\ref{prop:tract} and Prop.~\ref{prop:mul_regular}, we have that~\eqref{eq:indu_ti} is tractable if and only if the matrix pencil~$(A,B)^{(k)}$ is regular, or equivalently, if and only if at least one of the matrices~$A,B$ is regular. 
   
As a specific  example, consider
$    A = \begin{bmatrix}
        0&1\\0&0
    \end{bmatrix}$, and~$
    B = \begin{bmatrix}
        0&0\\1&0
    \end{bmatrix}.$
%%%%%%%%%%%%%%%%%%%%%%%
Then~$A^{(2)}=\det(A)=0$ and~$B^{(2)}=\det(B)=0$, so~$(A,B)^{(2)}=A^{(2)}-\lambda B^{(2)} = 0$. Note that in this case~$(A,B)$ is regular, and~$(A,B)^{(2)}$ is singular.

Important  dynamical  properties  of~\eqref{eq:indu_ti} follow from
combining Prop.~\ref{prop:x0_consist} and Prop.~\ref{prop:mul_gen_eig}. Suppose that the $k$-multiplicative compound DAE is tractable. Then there exists~$\lambda \in \C$ such that $\det(A^{(k)} - \lambda B^{(k)}) \neq 0$. Let
\begin{align*}
    \hat{B}_{k,\lambda}&:= (A^{(k)}-\lambda B ^{(k)} )^{-1} B^{(k)}, \\
    \hat{A}_{k,\lambda}&:= (A^{(k)}-\lambda B^{(k)}  )^{-1} A^{(k)} , 
\end{align*}
and let $i:=\ind (\hat{B}_{k,\lambda})$.
An initial condition $y(0)\in\R^{\binom{n}{k}}$ is~\emph{consistent}
iff $
y(0)$ is in the range of~$   (\hat{B}_{k,\lambda})^{i}  $,
and for such an initial condition the solution of~\eqref{eq:indu_ti} is
\begin{align}\label{eq:xAD_B_hat_indu}
    y(j ) = 
    \left ( 
    ( \hat{B}_{k,\lambda} )^{D}
    \hat{A}_{k,\lambda}
    \right )^{j}
    y(0), \quad 
    j=0,1, \dots.
\end{align}
%
%%%%%%%%%%%%%%%
Furthermore, if~$(A,B)$ has $s\geq k$ finite eigenvalues, denoted~$\lambda_i$, $i=1,\dots,s$,  then~\eqref{eq:indu_ti} is asymptotically stable iff~$
    \prod_{i=1}^k |\lambda_{\alpha_i}| < 1$  for all~$\alpha \in Q(k,s)$.

%%%%%%%%%%%%%%%%%%%%%%%%%%%%%%%%%%%%%%%%%%%%%%%%%%%%%%%%%%%%%
 
\subsection{Consistent and non-consistent initial conditions of the $k$-multiplicative compound DAE}
%%%%%%%%%%%%%%%%%%%
Consider the LTI DAE~\eqref{eq:DAE_TINV}. We already know that given~$k \in \{1,\dots,n\}$ solutions to~\eqref{eq:DAE_TINV},  the vector~$y(j)$  in~\eqref{eq:kcomp_y}  is a solution to the $k$-multiplicative compound DAE~\eqref{eq:indu_ti}. However,  Prop.~\ref{prop:mul_regular} implies that
under certain conditions a tractable DAE will induce a non-tractable $k$-multiplicative
compound DAE for any $k > 1$.

On the other hand, when $(A,B)^{(k)}$ is regular for all~$k \in \{1,\dots,n\}$ it is possible that for large values of~$k$ the only consistent initial condition of the~$k$-multiplicative compound system is zero. Indeed, for large~$k$,
Eq.~\eqref{eq:DAE_TINV} may not have~$k$ linearly-independent consistent initial conditions. The following results analyze these issues. 
We begin with the case where~$(A,B)^{(\ell)}$ is regular for all~$\ell \in \{1,\dots,n\}$.

\begin{Proposition}\label{prop:consistent_dim_k}
Suppose that~$(A,B)^{(\ell)}$ is regular for any~$\ell \geq  1$. Fix~$k \in \{1,\dots,n\}$. Let $\mathcal{V}^1$ denote the subspace of consistent initial conditions of~\eqref{eq:DAE_TINV}, and let~$\mathcal{V}^k$ denote the subspace of consistent initial conditions of the $k$-multiplicative compound DAE~\eqref{eq:indu_ti}. Then
\begin{equation}
    \dim(\mathcal{V}^k) = \binom{\dim(\mathcal{V}^1)}{k},
\end{equation}
where $\binom{\dim(\mathcal{V}^1)}{k}$ is defined to be zero for $k>\dim(\mathcal{V}^1)$.
\end{Proposition}

Prop.~\ref{prop:consistent_dim_k} implies in particular that~\eqref{eq:indu_ti} will have zero as its only consistent initial condition for any~$k>\dim(\mathcal{V}^1)$.

\begin{proof}
%%%%%%%%%%%%%%%%%%%%%
Let $s$ denote the number of finite eigenvalues of~$(A,B)$, counting multiplicities. Recall that $\dim(\mathcal{V}^1) = s$. It follows from Prop.~\ref{prop:mul_gen_eig} that $(A,B)^k$ has $\binom{s}{k}$ finite eigenvalues, and this completes the proof of Prop.~\ref{prop:consistent_dim_k}.
\end{proof}

Note that $\dim(\mathcal{V}^1) = \rank((\hat{B}_\lambda)^{\ind(\hat{B}_\lambda)}) \le \rank(B)$. Suppose that $\rank(B) < n$ and fix $k > \rank(B)$. Then $\hat{B}_{k,\lambda} = 0$, so $\dim(\mathcal{V}^k) = 0$. However, often $\dim(\mathcal{V}^1)$ is strictly smaller than $\rank(B)$, and then  there exists $k \le \rank(B)$ such that $\dim(\mathcal{V }^k) = 0$ but $\hat{B}_{k,\lambda} \neq 0$. In this case $\hat{B}_{k,\lambda}$ will be nilpotent. The next example illustrates  this.

\begin{Example}
Consider~\eqref{eq:indu_ti} with
\begin{equation*}
    A = \begin{bmatrix}
        -2 & -3 &  1 \\
         1 &  0 &  0 \\
         1 &  1 &  0
    \end{bmatrix}, \text{ and }
    B = \begin{bmatrix}
        1 & 0 & 0 \\
        0 & 1 & 0 \\
        0 & 0 & 0
    \end{bmatrix}.
\end{equation*}
Let~$\lambda = 1$, and note that $\det(A - B) \neq 0$. It may be verified that $\ind(\hat{B}_\lambda) = 2$ and $\rank(\hat{B}_\lambda^{\ind(\hat{B}_\lambda)})=1$, so $\dim(\mathcal{V}^1) =1$, and  Prop.~\ref{prop:consistent_dim_k} implies that~$\dim(\mathcal{V}^k) = 0$ for~$k=2,3$.

We now show directly that~$\dim(\mathcal{V}^2) = 0$. The~$2$-compound system  is
\be\label{eq:thyr2}
B^{(2)}y(j+1)=A^{(2)} y(j).
\ee 
 The matrix~$A^{(2)} - B^{(2)} $ is regular and multiplying~\eqref{eq:thyr2}
 on the left by~$(A^{(2)} - B^{(2)})^{-1}$ gives~$\hat{B}_{2,1}y(j+1)=\hat{A}_{2,1} y(j)$,
that is, 
\[
\begin{bmatrix}
        0 & 0 & 0 \\
        -1 & 0 & 0 \\
        1 & 0 & 0
    \end{bmatrix} y(j+1)= \begin{bmatrix}
    1&0&0\\-1&1&0\\1&0&1
\end{bmatrix}  y(j).
\]
The first equation here gives~$y_1(j)\equiv 0$. Using this in the second equation gives~$y_2(j)\equiv 0$, and now the third equation gives~$y_3(j)\equiv 0$, so indeed the only consistent initial condition is~$y(0)=0$. 
 Note that  the matrix
$
    \hat{B}_{2,1} = \begin{bmatrix}
        0 & 0 & 0 \\
        -1 & 0 & 0 \\
        1 & 0 & 0
    \end{bmatrix} 
$
  is nilpotent, as expected.
\end{Example}

We now turn to consider the case where~$(A,B)$ is regular,
but $(A,B)^{(k)}$ is singular for any~$k>1$. The following result shows that in this case the $k$-compound system will have a consistent non-zero   initial condition for any~$k$. 

\begin{Proposition}\label{prop:induced_singular}
    Let $A,B \in \R^{n \times n}$ be such that $(A,B)$ is regular and $\det(A)=\det(B)=0$. 
    Fix $k>1$. Then
    there exists a vector $z \in \R^{\binom{n}{k}} \setminus \{0\}$ such that:
    \begin{enumerate}
        \item $z$ is a consistent initial condition for the $k$-multiplicative
        compound DAE~\eqref{eq:indu_ti}; \label{item:cdro}
        \item if $y(j)$, $j=0,1,\dots$ is a solution of~\eqref{eq:indu_ti} then $y(j) + z$, $j=0,1,\dots$, is another solution of~\eqref{eq:indu_ti} for the initial condition $y(0)$. \label{item:addsol}
    \end{enumerate}
\end{Proposition}

\begin{proof}
%%%%%%%%%%%%%%%%%%%%
Since $\det(A) = \det(B) = 0$,  Prop.~\ref{prop:mul_regular}   implies
that there exists
$z \in \ker(A^{(k)}) \cap \ker(B^{(k)})$
with~$z\not =0$.
Consider the sequence~$ {y}(j) \equiv  z$ for all $j \ge 0$. Since $B^{(k)}  {y}(j+1) = A^{(k)} {y}(j) = 0$, $y(j)$ is a solution of the~$k$-compound system and, in particular, the vector~$z$ is indeed a consistent initial condition of~\eqref{eq:indu_ti}. This proves~\eqref{item:cdro}. The proof of~\eqref{item:addsol} follows similarly, and this completes the proof of Prop.~\ref{prop:induced_singular}.
%%%%%%%%%%%%
%%%%%%%%%%%%%
\end{proof}

 The singularity of~$(A,B)^{(k)}$ implies that the~$k$-compound system may have consistent initial conditions  and solutions that do not correspond to 
 $k$-compounds of consistent  initial  conditions and solutions of the original system.
The next example illustrates this. 
%%%%%%%%%%%%%%%%
\begin{Example}\label{exa:k_comp_singular}
Consider~\eqref{eq:indu_ti} with~$A=\diag(0,1/2,1)$ and~$B=\diag(1,1,0)$,
that is,
\begin{align}\label{eq:abc} 
    x_1(j+1) &= 0, \nonumber\\
    x_2(j+1) &=  x_2(j)/2, \nonumber\\
    0 &= x_3(j).
\end{align}
Note that the pencil~$(A,B)$ is regular, but since~$\det(A)=\det(B)=0$, the pencil~$(A,B)^{(k)} $ is singular for any~$k>1$. The subspace of consistent initial conditions of~\eqref{eq:abc} is
\begin{equation}\label{eq:exa_singular_consistent}
\mathcal{V}^1=    \spanop (e^1,e^2),
\end{equation}
where~$e^i$ is the~$i$th canonical vector in~$\R^3$. 
Furthermore, given a consistent initial condition~$a \in \mathcal{V}^1$, the corresponding solution is~$x(0)=a$ and
$
    x(j) = 2^{-j} a_2 e^2 $, $ j=1,2,\dots$.

Consider now the~$2$-compound system. Since~$A^{(2)}=\diag(0,0,1/2)$ and~$B^{(2)}=\diag(1,0,0)$, the $2$-compound system is  
\begin{align}\label{eq:exa_compound_singular}
    y_1(j+1) &= 0, \nonumber \\
    0 &= 0, \nonumber \\
    0 &=  y_3(j)/2.
\end{align}
This implies that~$\mathcal{V}^2=\mathcal{V}^1$. 
For any~$a, b \in\mathcal{V}^1$, we have  
$
    \begin{bmatrix}
        a  & b 
    \end{bmatrix}^{(2)} = \begin{bmatrix}
       a_1 b_2-b_1 a_2 &
        0 &
        0
    \end{bmatrix}^T.
$
Thus, the~$2$-compound system
has consistent initial conditions that are \emph{not}  $2$-compounds of  consistent initial conditions of the original system. Furthermore,  it is easy to see that the~$2$-compound system has solutions that are \emph{not} $2$-compounds of solutions of the original system. 
%%%%%%%%%%%%%%%%%%%%%%%%%%%%%%%%%%%%%
\end{Example}

\section{An Application:  the set of stable initial conditions of a~DAE}\label{sec:appli}
%%%%%%%%%%%%%%%%%%%%%%%

We now use our results to analyse stable subspaces of a DAE by studying the stability of the $k$-compound DAE. This is inspired by \cite{muldowney1990compound}, which related the subspace of initial conditions of the linear time-varying ODE~$\dot x(t)=A(t)x(t)$ which lead to an asymptotically stable solution to the $k$-compound ODE. Here we rephrase this result in terms of discrete-time systems and generalize it to linear time-varying~DAEs.

We require the following definition. The  linear time-varying DAE~\eqref{eq:DAE_TINV}
is called   \emph{uniformly stable} if for any~$\epsilon > 0$ there exists~$\delta > 0$ such that for any~$j_0 \in \R$, if~$x(j_0)$ is a consistent initial condition and~$|x(j_0)| < \delta$, then~$|x(j)| < \epsilon$ for all~$j \ge j_0$.
%%%%%%%%%%%%%%%%%%%%%%%%%%%%%%%%%%%%%%%%%%%%%%
\begin{Theorem}\label{thm:subspace}
%%%%%%%%%%%%%%%%%%%%%%%%%%%%%%%%%%%%%%%
    Suppose that the time-varying DAE~\eqref{eq:lin_dae_dt} is tractable and uniformly stable. Fix an initial time~$j_0$, and let~$\mathcal{V}^1(j_0)$ denote the set of consistent initial conditions. Fix~$k \in \{1,\dots,\dim(\mathcal{V}^1(j_0))\}$. Consider the following assertions.
    \begin{enumerate}[(a)]
    %%%%%%%%%%%%%%%%%%%%%%%%%%%%%%%%%%
    \item \label{cond:stab_k_comp}
        Every solution of the~$k$-multiplicative compound DAE~\eqref{eq:compound_dae_dt}
        satisfies~$\lim_{j \to \infty} y(j) = 0$;
    \item \label{cond:stable_subspace}
        The DAE~\eqref{eq:lin_dae_dt}  admits a subspace~$\mathcal{X}(j_0) \subseteq \mathcal{V}^1(j_0)$, with $\dim(\mathcal{X}(j_0)) = \dim(\mathcal{V}^1(j_0)) - k + 1$, such that \be\label{eq:limit_cond}
        %%%
        \lim_{j \to \infty} x(j) = 0 \text{ for any } x(j_0) \in \mathcal{X}(j_0).
        \ee
    \end{enumerate}
    Then~\ref{cond:stab_k_comp} implies~\ref{cond:stable_subspace}. 
\end{Theorem}
%%%%%%%%%%%%%%%%%%%
\begin{proof}
    Suppose that every solution of~\eqref{eq:compound_dae_dt} satisfies
    ${\lim_{j \to \infty} y(j) = 0}$. Pick~$k$ vectors $a^1,\dots,a^k \in \mathcal{V}^1$. 
    Define~$X(j) := \begin{bmatrix} x(j,a^1) & \dots & x(j,a^k) \end{bmatrix}$.  Since uniform stability implies that all trajectories are bounded (by a constant which depends on the initial condition), there exists an increasing sequence of times~$j_i$ such that $\lim_{i \to \infty} j_i = \infty$ and~$P := \lim_{i \to \infty} X(j_i)$ exists. By Prop.~\ref{prop:compound_dae_dt},
    \begin{equation}
        B^{(k)}(j+1) X^{(k)}(j+1) = A^{(k)}(j) X^{(k)}(j),
    \end{equation}
    so~$P^{(k)} = 0$, and by~\eqref{eq:k_mul_rank}~$P$ has non-full rank, so there exists~$c \in \R^{k}\setminus\{0\}$ such that~$Pc=0$, that is, 
    \begin{align*}
        0  &= \lim_{i \to \infty} \sum_{\ell=1}^k c_\ell x(j_i, a^\ell) \\
          &= \lim_{i \to \infty} x(j_i, \sum_{\ell=1}^k c_\ell a^\ell)\\
          &= \lim_{j \to \infty} x(j, \sum_{\ell=1}^k c_\ell a^\ell),
    \end{align*}
    where the last step follows from uniform stability. Note that we may choose~$a^1,\dots,a^k$ arbitrarily as long as they are consistent initial conditions. In particular, we may choose them to be linearly independent. Therefore, the subspace of consistent initial conditions which yield trajectories which are not asymptotically stable has dimension of at most $k-1$, and this completes the proof.
\end{proof}

Note that unlike the case for ODEs, the existence of a stable set of initial conditions with dimension $\dim(\mathcal{V}^1) - k + 1$ does not imply that the $k$-compound system is asymptotically stable. This is due to the fact that the $k$-compound system might have solutions which do not correspond to compounds of solutions of the original system.
\begin{Example}
    Consider again the system from Example~\ref{exa:k_comp_singular}. It is easy to verify that all solutions converge to the origin asymptotically, so the system is uniformly stable, and we may take~$\mathcal{X} = \mathcal{V}^1$. However, Prop.~\ref{prop:induced_singular} implies that the~$k$-compound system will have a constant non-zero solution, so it is not asymptotically stable.
\end{Example}

%%%%%%%%%%%%%%%%%%%%%%%%%%%%%%%%%%%%%%%%%%%%%%%%%%%%%%%
\subsection{Application to a 3D Leslie model}
%%%%%%%%%%%%%%%%%%%%%%%%%%%%%%%%%%%%%%%%
We describe an  application of Thm.~\ref{thm:subspace}
to the Leslie model from mathematical demography (see, e.g.,~\cite[Ch.~22]{math_eco_Kot}).
Consider the system: 
\be\label{eq:zles}
z(j+1)=L  z(j),
\ee
with~$z\in\R^3$ and~$L=\begin{bmatrix}
    b_1 & b_2 & 0 \\
    p_1 & 0   & 0 \\
    0   & p_2 & 0
\end{bmatrix}$. The parameters~$b_i>0$ [$p_i>0$] represent 
age-class fertilities [age-class
survival probabilities].
 Suppose that we can measure the population
at the current time, denoted~$\ell$, 
and we are interested in projecting the  population dynamics backwards in time~\cite[Ch.~9]{Campbell_and_Meyer}. Letting~$x(0):=z(\ell)$, $x(1):=z(\ell-1) $, and so on, 
gives~$Lx(1)=Lz(\ell-1)=z(\ell)=x(0)$, so we are led to consider   
 the~DAE 
\be\label{eq:lde_leslie}
Lx(j+1)=x(j). 
\ee 
The corresponding matrix pencil~$( I_3,L)$  is regular, so Prop.~\ref{prop:tract} implies that~\eqref{eq:lde_leslie} is tractable. 
We have~$\ind(L)=1$, and applying 
Prop.~\ref{prop:x0_consist} with~$\lambda=0$ implies that the set of consistent initial conditions is~$\mathcal{V}^1=\spanop\{\begin{bmatrix} b_1&p_1&0
 \end{bmatrix}^T,\begin{bmatrix} b_2&0&p_2
 \end{bmatrix}^T
\}$, 
and for any~$x(0)\in \mathcal{V}^1 $ the unique 
solution of~\eqref{eq:lde_leslie}  is 
\begin{align}\label{eq:LDE_LESL_SOL} 
    x(j ) = (L^D)^j 
    x(0), \quad 
    j=0,1, \dots.
\end{align}
A calculation gives 
$L^D= \begin{bmatrix}
0&p_1^{-1} &0\\
b_2^{-1} & -b_1b_2^{-1}p_1^{-1} & 0\\
-b_1 p_2 b_2^{-2}p_1^{-1} & c&0
\end{bmatrix}$,
with~$c:=  \frac{b_1^2p_2}{b_2^2p_1^2}+\frac{p_2}{b_2p_1} $.

The~$2$-compound system~\eqref{eq:indu_ti} is~$L^{(2)}y(j+1)=y(j)$, with
$
L^{(2)} = \begin{bmatrix}
     -b_2 p_1 & 0 & 0 \\
      b_1 p_2 & 0 & 0 \\
      p_1 p_2 & 0 & 0 
\end{bmatrix} .
$
Note in particular that the $2$-compound system has a simpler structure than the original system, and is therefore easier to analyse. Note also that this simple structure would still hold even if the system were time-varying, i.e. if~$b_i,p_i$ vary with the time~$j$.

Assume that $b_2 p_1 > 1$.
Then condition~\ref{cond:stab_k_comp} holds, so 
Thm.~\ref{thm:subspace}
implies that the
DAE~\eqref{eq:lde_leslie} admits a one-dimensional set of initial conditions~$\mathcal{X}$
such that~\eqref{eq:limit_cond} holds. Indeed, it can be shown that one of the finite eigenvalues of~$( I_3, L )$ (or equivalently of $L^D$) is $2/(b_1 + \sqrt{b_1^2 + 4b_2p_1}) $, which is smaller than one when $b_2 p_1 > 1$. Therefore,~\eqref{eq:lde_leslie} has a one-dimensional subspace of initial conditions for which the corresponding solutions converge to the origin. It can also be shown that this stable eigenvalue has a corresponding eigenvector with positive entries, consistent with the fact that we expect the state-variables (that represent
populations) to be non-negative.
 To explain this result, note that~$b_2 p_1$ is the ratio of current children to future offspring: $p_1$ describes the proportion of children which survive to the second age group, and~$b_2$ describes the fertility of the surviving individuals. Intuitively, the condition~$b_2 p_1 > 1$  implies that the population in~\eqref{eq:zles} grows with (forward) time, so the population in~\eqref{eq:lde_leslie} decreases.
 
Fig.~\ref{fig:leslie_sim} shows two trajectories of~\eqref{eq:lde_leslie} with the parameters $p_1 = 0.9, p_2 = 0.7, b_1 = 1.1$ and $b_2 = 2.3$ (so~$b_2 p_1 > 1$), projected for simplicity onto the 2-dimensional space $\mathcal{V}^1$ using the orthogonal projection matrix  derived by applying the Gram-Schmidt process. The figure shows convergence along one direction (clearly demonstrated by the trajectory with circle markers), and diverging oscillations along a different direction. Fig.~\ref{fig:leslie_sim} also shows the corresponding solution to the 2-compound system, which converges to the origin asymptotically.
%%%%
%%%%%

\begin{figure}
    \centering
    \begin{tikzpicture}[scale=0.475]
    \begin{groupplot}[group style={rows=1,columns=2,horizontal sep=1.8cm},scale only axis,width=0.8\columnwidth]
        \nextgroupplot[
            xlabel = {$(Px)_1$},
            ylabel = {$(Px)_2$},
            label style={font=\Large},
            tick label style={font=\large},
            yticklabel style = {
                /pgf/number format/fixed,
                /pgf/number format/precision=2
            },
            xmin = 0,
        ]
            \addplot[thick,mark=asterisk] coordinates {(0.4765,0.1515)};
            \addplot[thick,dashed,mark=none] table {
                0.4765    0.1515
                0.2280    0.0725
            };
            \addplot[thick,dashed,mark=*,mark options={solid,fill=white}] table {
                0.2280    0.0725
                0.1091    0.0347
                0.0522    0.0166
                0.0250    0.0079
            };
        
            \addplot[thick,mark=asterisk] coordinates {(0.5,0.14)};
            \addplot[thick,dotted,mark=none] table {
                0.5000    0.1400
                0.2432    0.0965
            };
            \addplot[thick,dotted,mark=diamond*,mark options={solid,fill=black},mark size=3pt] table {
                0.2432    0.0965
                0.1123    0.0163
                0.0578    0.0380
                0.0236   -0.0123
            };
        \nextgroupplot[
            xmin = 1,
            xmax = 5,
            ymin = 0,
            xlabel = {$j$},
            ylabel = {$P^{(2)}y(j)$},
            label style={font=\Large},
            tick label style={font=\large},
            yticklabel style = {
                /pgf/number format/fixed,
                /pgf/number format/precision=2
            },
            ylabel shift = -3pt
        ]
            \addplot[thick,dashed,mark=square*,mark options={solid,fill=white},mark size=3pt] table {
                1 0.0091
                2 0.0044
                3 0.0021
                4 0.0010
                5 0.0005
            };
    \end{groupplot}
    \end{tikzpicture}
    \caption{\textbf{Left:} two trajectories of the Leslie model~\eqref{eq:lde_leslie}, projected onto the 2-dimensional space $\mathcal{V}^1$ using a projection matrix~$P$. Initial values are shown with asterisks ($*$). \textbf{Right:} corresponding solution of the 2-compound system, projected onto the one-dimensional space $\mathcal{V}^2$ (see the definition  in Prop.~\ref{prop:consistent_dim_k}). }
    \label{fig:leslie_sim}
\end{figure}
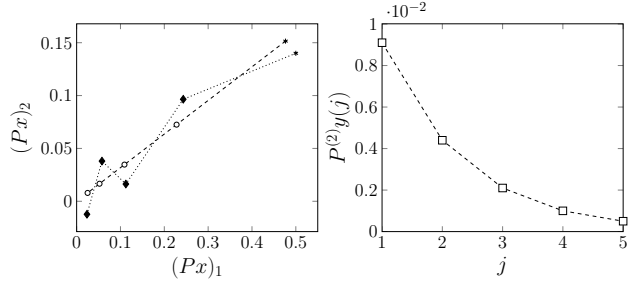

\section{Conclusion}
%%%%%%%%%%%%%%%%%%%%%
Given square matrices~$A,B$,
we defined the  $k$-multiplicative compound of the matrix pencil~$(A,B)$. This is a matrix pencil, denoted~$(A,B)^{(k)}$, that for~$k=1$   reduces to~$(A,B)$.
We studied the relation between~$(A,B)$ and~$(A,B)^{(k)}$ and illustrated several applications to~DAEs.
In particular, we showed that the 
DAE corresponding to~$(A,B)^{(k)}$
describes the evolution of~$k$-parallelotopes in the DAE corresponding to~$(A,B)$.

An interesting  line of research is defining also the~\emph{$k$-additive} compound of a matrix pencil, and using it to analyze differential-algebraic equations. 

The~$k$-compounds of a matrix have been recently used to define non-trivial generalizations of several classes  of both continuous-time and discrete-time 
dynamical systems including $k$-positive linear systems and 
$k$-cooperative nonlinear systems~\citep{Eyal_k_posi}, $k$-contracting systems~\citep{kordercont}, 
$k$-diagonally stable systems~\citep{cheng_diag_stab},  
and more. Another possible research  direction 
 is to use the compounds of matrix pencils to define   such generalizations for   difference-algebraic and differential-algebraic systems.

\section*{Appendix: Drazin inverse of the $k$-multiplicative compound}
%%%%%%%%%%%%%%%%%%%%%%%%%%%%%%%%%%%%%%%%%%%%%%%%%%%%%%
Eq.~\eqref{eq:xAD_B_hat_indu} includes  the Drazin inverse of the $k$-multiplicative compound of a matrix. The next result shows that this is equal to the  $k$-multiplicative of the Drazin inverse of the original matrix. 
%%%%%%%%%%%%%%%%%%%%%%%%%%%%%%%%
\begin{Proposition}
\label{thm:k_mul_drazin}
    Let $A \in \C^{n \times n}$, and fix $k \in \{1,\dots,n\}$. 
    Then~$ (A^{(k)})^D = (A^D)^{(k)}$.
\end{Proposition}
%%%
 %%%%%%%%%%%%%%%%%%
\begin{proof}
%%%%%%%%%%%%%%%
Denote~$i:=\ind(A)$, $j:=\ind(A^{(k)})$,  and~$E:=A^{(k)}$. We need to show that~$E^D=  (A^D)^{(k)} $.
Since~$A^D$ is the Drazin inverse of~$A$, we have 
\begin{align}\label{eq:pout}
 A^{i+1}A^D=A^i, \;
 A A^D = A^D A, \; 
 A^D A A^D  = A^D.
 \end{align} 
 %%%%%
Taking the~$k$-multiplicative compounds of these equations and using the 
Cauchy-Binet Theorem gives
\begin{align}\label{eq:rtp}
    E^{i+1} (A^D)^{(k)} &= E^i, \\
    E (A^D) ^{(k)} &=(A^D) ^{(k)}  E,\nonumber\\
    (A^D)^{(k)}  E  (A^D)^{(k)}  &=( A^D)^{(k)}\nonumber.
\end{align}
Thus, $(A^D) ^{(k)}$ satisfies two of the requirements for the Drazin inverse of~$E$, and we only  need to show that
\be\label{eq:enou}
 E^{j+1} (A^D)^{(k)} = E^j . 
\ee
Let~$A$ have the Jordan decomposition in~\eqref{eq:jordan_A}. Then the index of nilpotency of~$N$ is also~$i$, that is,~$i$ is the minimal integer such that~$N^i=0$. 
We prove the proposition when~$T=I$, so~$A=\diag(C,N)  $,  
 $A^i=\diag(
C^i,0)$,
and~$A^D=\diag(C^{-1}, 0) $.  The proof in the more general case is very similar. 
Denote the dimension of~$C$ by~$s$. Then~$\rank(A^i)=s$.
We consider two cases. 

\noindent \emph{Case 1:} Assume that~$k>s$. 
Then every~$(k\times k)$-submatrix of~$A^D$ includes either a column of zeros or a row of zeros, so~$(A^D)^{(k)}=0$. Also, since every eigenvalue of~$A^{(k)}$
is the product of~$k$ eigenvalues of~$A$,     every eigenvalue of~$E$ is zero. Thus,~$E$ is nilpotent, so~$E^j=0$. We conclude  that~\eqref{eq:enou} holds. 

\noindent \emph{Case 2:} Assume that~$k\leq s$. 
We will show that in this case~$j=i$. Fix an integer~$\ell\geq 0 $. Then 
 \begin{align*}
     \rank(E^{\ell+1})&= \rank( (A^{(k)})^{\ell+1} )\\
     &=\rank( (A^{\ell+1} )^{(k)} ) \\
     &= \binom{\rank ( A^{\ell+1}) }{k},
 \end{align*}
where the last equation follows from~\eqref{eq:k_mul_rank}. Combining this with the definition of~$i$ gives
 \begin{align}\label{eq:ertp}
\rank(E^{i+1})&=\binom{\rank ( A^{i+1}) }{k}\nonumber\\
&=\binom{\rank ( A^{i }) }{k}\nonumber\\
&=\rank(E^{i}).
 \end{align}
Also, for any~$p<i$ we have~$A^p=\begin{bmatrix}
C^p&0\\0&N^p
\end{bmatrix}$, with~$C \in \C^{s\times s}$ and~$N^p\not =0$, and combining this with~$k\leq s$ gives
 \begin{align}\label{eq:pp1}
\rank(E^{p+1})&=\binom{\rank ( A^{p+1}) }{k}\nonumber\\
&<\binom{\rank ( A^{p }) }{k}\nonumber\\
&=\rank(E^{p}).
 \end{align}
Combining~\eqref{eq:ertp} and~\eqref{eq:pp1} proves that~$j=i$, and thus~\eqref{eq:rtp} implies~\eqref{eq:enou}. This completes the proof of Prop.~\ref{thm:k_mul_drazin}.
%%%%
\end{proof}

\begin{Example}
Let~$
A=\diag(a_1,\dots,a_s,0,\dots,0)\in\R^{n\times n},
$
with~$a_i\not =0$. Fix~$k\leq s$. Then on the one-hand~$A^{(k)}=\diag(  \prod_{i=1}^k a_i, \dots,\prod_{i=s-k+1}^s a_i, 0,\dots,0  )$, so
$
(A^{(k)})^D=\diag(  \prod_{i=1}^k a_i^{-1}, \dots,\prod_{i=s-k+1}^s a_i^{-1}, 0,\dots,0  ).
$
On the other-hand, $A^D=\diag(a_1^{-1},\dots,a_s^{-1},0,\dots,0)$
and thus
\[
(A^D)^{(k)}=\diag\left(  \prod_{i=1}^k a_i^{-1}, \dots,\prod_{i=s-k+1}^s a_i^{-1}, 0,\dots,0  \right),
\]
so~$(A^D)^{(k)}=(A^{(k)})^D$.
%%%%%
\end{Example}

Note that if~$A$ is regular then~$A^D=A^{-1}$,
so Prop.~\ref{thm:k_mul_drazin} reduces to  the well-known relation~$ 
        (A^{(k)})^{-1} = (A^{-1})^{(k)}$.

  \section*{Acknowledgements}
  %%%%%%%%%%%%%%%%%%%%%%%%%%%%%%%%%%%%%%%%%%%%%%%%%%%%%%%%%%%%%%%%
  We are grateful to E. D. Sontag for helpful comments.  We thank the anonymous reviewers 
  and the editor for many helpful comments.

\bibliographystyle{abbrvnat}
\bibliography{literature}

%%%%%%%%%%%%%%%%%%
 \subsection*{  } % This subsection (with no heading) is added to give more space between two biographies
    \setlength\intextsep{0pt} % align top of photo with text
   % \begin{wrapfigure}{l}{0.13\textwidth}
%        \centering
%        \includegraphics[width=0.15\textwidth] {}
%    \end{wrapfigure}
    \noindent \textbf{Ron Ofir} %(Student Member, IEEE) 
    received his BSc degree (cum laude) in Elec.  Eng. from the Technion-Israel Institute of Technology in~2019. He is currently pursuing his Ph.D. degree in the  Andrew and Erna Viterbi Faculty of Electrical and Computer Engineering, Technion - Israel Institute of Technology. His current research interests include compound matrices and
 contraction theory with
 applications in  dynamics and control of power systems.

 \subsection*{  } % This subsection (with no heading) is added to give more space between two biographies
    \setlength\intextsep{0pt} % align top of photo with text
    \begin{wrapfigure}{l}{0.13\textwidth}
        \centering
        \includegraphics[width=0.15\textwidth]{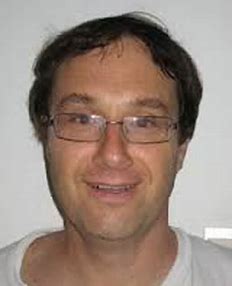}
    \end{wrapfigure}
    \noindent \textbf{Michael Margaliot}  received the BSc (cum laude) and MSc degrees in
 Elec. Eng. from the Technion-Israel Institute of Technology-in
 1992 and 1995, respectively, and the PhD degree (summa cum laude) from Tel
 Aviv University in 1999. He was a post-doctoral fellow in the Dept. of
 Theoretical Math. at the Weizmann Institute of Science. In 2000, he
 joined the Dept. of Elec. Eng.-Systems, Tel Aviv University,
 where he is currently a Professor. His  research
 interests include the stability analysis of differential inclusions and
 switched systems, optimal control theory, computation with
 words, Boolean control networks, contraction theory,  applications of matrix compounds in systems and control theory, and systems biology.
 He is co-author of \emph{New Approaches to Fuzzy Modeling and Control: Design and
 Analysis}, World Scientific,~2000 and of \emph{Knowledge-Based Neurocomputing}, Springer,~2009. 
 He  served
as  an Associate Editor of~\emph{IEEE Trans. on Automatic Control} during 2015-2017.

\end{document}